\documentclass{article}
\usepackage[a-1b]{pdfx} 
\usepackage[utf8]{inputenc}
\usepackage{amsmath}
\usepackage{amsthm}
\usepackage{amsfonts}
\usepackage{graphicx}
\usepackage{xcolor}
\usepackage{thmtools, thm-restate}
\usepackage{fullpage}

\newtheorem{conjecture}{Conjecture}

\newtheorem{theorem}{Theorem}
\newtheorem{lemma}{Lemma}
\newtheorem{observation}{Observation}

\title{Polygons with Prescribed Angles in 2D and 3D\thanks{Research on this paper is supported, in part, by NSF grants CCF-1740858, CCF-1712119, DMS-1800734, and DMS-1839274.}}
\author{Alon Efrat, Radoslav Fulek, Stephen Kobourov, Csaba D. T\'{o}th}

\date{}

\begin{document}

\maketitle

\begin{abstract}
We consider the construction of a polygon $P$ with $n$ vertices whose turning angles at the vertices are given by a sequence $A=(\alpha_0,\ldots, \alpha_{n-1})$,  $\alpha_i\in (-\pi,\pi)$, 
for $i\in\{0,\ldots, n-1\}$.
The problem of realizing $A$ by a polygon can be seen as that of constructing a straight-line drawing of a graph with prescribed angles at vertices, and hence, it is a special case of the well studied problem of constructing an \emph{angle graph}.
In 2D, we characterize sequences $A$ for which every generic polygon $P\subset \mathbb{R}^2$ realizing $A$ has at least $c$ crossings, for every $c\in \mathbb{N}$, and describe an efficient algorithm that constructs, for a given sequence $A$, a generic polygon $P\subset \mathbb{R}^2$ that realizes $A$ with the minimum number of crossings.
In 3D, we describe an efficient algorithm that tests whether a given sequence $A$ can be realized by a (not necessarily generic) polygon $P\subset \mathbb{R}^3$, and for every realizable sequence the algorithm finds a realization.
 \end{abstract}

--------------------------------------------------

\section{Introduction}
\label{sec:intro}

Straight-line realizations of graphs with given metric properties have been one of the earliest applications of graph theory. Rigidity theory, for example, studies realizations of graphs with prescribed edge lengths, but also considers a mixed model
where the edges have prescribed lengths or directions~\cite{CJK20,JacksonJ10,JacksonK11a,JohnS09,SaliolaW04}.
In this paper, we extend research on the so-called \emph{angle graphs}, introduced by Vijayan~\cite{vijayan1986geometry} in the 1980s, which are geometric graphs with prescribed angles between adjacent edges. Angle graphs found applications in mesh flattening~\cite{zayer_etal:05}, and computation of conformal transformations~\cite{driscoll1998numerical,S99_crossratios} with applications in the theory of minimal surfaces and fluid dynamics.

Viyajan~\cite{vijayan1986geometry} characterized planar angle graphs under various constraints, including the case when the graph is a cycle~\cite[Theorem~2]{vijayan1986geometry} and when the graph is 2-connected~\cite[Theorem~3]{vijayan1986geometry}. In both cases,
the characterization leads to an efficient algorithm to find a
planar straight-line drawing or report that none exists.
Di Battista and Vismara~\cite{di1996angles} showed that for 3-connected angle graphs (e.g., a  triangulation), planarity testing reduces to solving a system of linear equations and inequalities in linear time. Garg~\cite{G98_anglegraphs} proved that planarity testing for angle graphs is NP-hard, disproving a conjecture by Viyajan. Bekos et
al.~\cite{BekosFK19} showed that the problem remains NP-hard even if
all angles are multiples of $\pi/4$.

The problem of computing (straight-line) realizations of angle graphs can be  seen as the problem of reconstructing a drawing of a graph from some given partial information. The research problems to decide if the given data uniquely determine the realization or its  parameters of interest are already interesting for cycles, and were previously considered in the areas of conformal transformations~\cite{S99_crossratios} and visibility graphs~\cite{dmw-11-pda}.

In 2D, we are concerned  with realizations of angle cycles as polygons minimizing the number of crossings which, as we shall see, depends only on the sum of the turning angles. It follows from the seminal work of Tutte~\cite{tutte1963draw} and Thomassen~\cite{thomassen1980planarity}  that every positive instance of a 3-connected planar angle graph admits a crossing-free realization if the prescribed angles yield convex faces. Convexity will also play a crucial role in our proofs. 

In 3D, we would like to determine whether a given angle cycle can be realized by a 
polygon. Somewhat counter-intuitively, self-intersections cannot be always avoided in a polygon realizing the given angle cycle in 3D; we present examples below. Di Battista et al.~\cite{DiBattistaKLLW12} characterized oriented polygons that can be realized in $\mathbb{R}^3$ without self-intersections with axis-parallel edges of given directions. Patrignani~\cite{Patrignani08} showed that recognizing crossing-free realizability is NP-hard for graphs of maximum degree 6 in this setting.  

Throughout the paper we assume modulo $n$ arithmetic on the indices, and use $\langle .,.\rangle$ scalar product notation.

\paragraph{\bf Angle sequences in 2-space.} In the plane, an \emph{angle sequence} $A$ is a sequence $(\alpha_0,\ldots,\alpha_{n-1})$ of real numbers such that $\alpha_i\in(-\pi,\pi)$ for all $i\in \{0,\ldots, n-1\}$. Let $P\subset \mathbb{R}^2$ be an oriented polygon with $n$ vertices $v_0,\ldots ,v_{n-1}$ that appear in the given order along $P$, which is consistent with the given orientation of $P$.
The \emph{turning angle} of $P$ at $v_i$ is the angle in $(-\pi,\pi)$ between the vector $v_i-v_{i-1}$ and $v_{i+1}-v_i$. The sign of the angle is positive if a rotation of the plane that maps the vector $v_i-v_{i-1}$ to the positive direction of the $x$-axis, makes the $y$-coordinate of $v_{i+1}-v_i$ positive. Otherwise, the angle nonpositive; see Fig.~\ref{fig:turningangle}.

\begin{figure}[htb]
\centering
\includegraphics[scale=1]{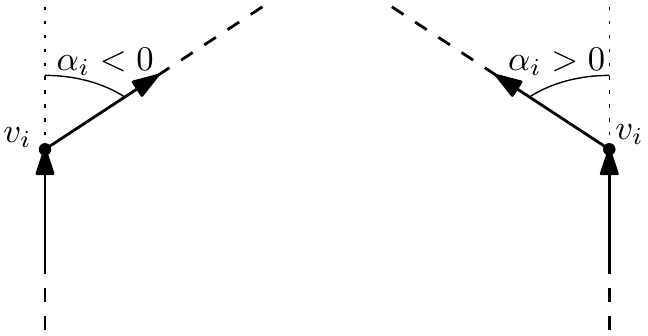}
\caption{A negative, or right, (on the left) and a positive, or left, (on the right) turning angle $\alpha_i$ at the vertex $v_i$ of an oriented polygon.}
\label{fig:turningangle}
\end{figure}

The oriented polygon $P$ \emph{realizes} the angle sequence $A$ if  the turning angle of $P$ at $v_i$ is equal to $\alpha_i$, for $i=0,\ldots, n-1$.
A polygon $P\subset \mathbb{R}^2$  is \emph{generic} if all its self-intersections  are transversal (that is, proper crossings), vertices of $P$ are distinct points, and no vertex of $P$ is contained in a relative interior of an edge of $P$. 
Following the terminology of Viyajan~\cite{vijayan1986geometry}, an
angle sequence $A=(\alpha_0,\ldots,\alpha_{n-1})$ is \emph{consistent} if there exists a generic polygon $P$ with $n$ vertices realizing $A$. For a polygon $P$ that realizes an angle sequence $A=(\alpha_0,\ldots,\alpha_{n-1})$ in the plane, 
the \emph{total curvature} of $P$ is $\mathrm{TC}(P)=\sum_{i=0}^{n-1}\alpha_i$,
and the \emph{turning number} (also known as \emph{rotation number}) of $P$ is $\mathrm{tn}(P)=\mathrm{TC}(P)/(2\pi)$,
where $\mathrm{tn}(P)\in \mathbb{Z}$~\cite{Sullivan2008}. Therefore a necessary condition for the consistency of an angle sequence is that $\sum_{i=0}^{n-1}\alpha_i \equiv 0 \pmod{2\pi}$. This condition is also sufficient except for the case when $\sum_{i=0}^{n-1}\alpha_i =0$.
We give a sufficient condition in all cases in the next paragraph to complete the characterization of consistent angle sequences.

Let $\beta_i=\sum_{j=0}^{i}\alpha_j \mod{2\pi}$, and let $\mathbf{u}_i\in \mathbb{R}^2$ be the unit vector $(\cos \beta_i, \sin \beta_i)$ for $i=0,\ldots , n-1$. 
As observed by Garg~\cite[Section~6]{G98_anglegraphs}, $A$ is consistent if and only if $\sum_{i=0}^{n-1}\alpha_i \equiv 0 \pmod{2\pi}$ and $\mathbf{0}$ is a strictly positive convex combination of vectors $\mathbf{u}_i$, that is, there exist scalars $\lambda_0,\ldots, \lambda_{n-1}> 0$ such that $\sum_{i=0}^{n-1}\lambda\mathbf{u}_i=\mathbf{0}$ and $\sum_{i=0}^{n-1}\lambda_i=1$. We use this characterization, in the proof of Theorem~\ref{thm:crossings_number} stated below.

The \emph{crossing number}, denoted by $\mathrm{cr}(P)$, of a generic polygon is the number of self-crossings of $P$. The \emph{crossing number} of a consistent angle sequence $A$ is the minimum integer $c$, denoted by $\mathrm{cr}(A)$, such that there exists a generic polygon $P\in \mathbb{R}^2$ realizing $A$ with $\mathrm{cr}(P)=c$. Our first main results is the following theorem.

\begin{restatable}{theorem}{CrossingNumber}
\label{thm:crossings_number}
For a consistent angle sequence $A=(\alpha_0,\ldots,\alpha_{n-1})$ in the plane, we have
\[
\mathrm{cr}(A)
 = \begin{cases}
    1       & \text{ if } \sum_{i=0}^{n-1}\alpha_i=0, \\
    |k|-1   & \text{ if } \sum_{i=0}^{n-1}\alpha_i=2k\pi \text{ and }k\neq 0.
           \end{cases}
\]
\end{restatable}

The proof of Theorem~\ref{thm:crossings_number} can be easily converted into a weakly linear-time algorithm that constructs, for a given consistent sequence $A$, a generic polygon $P\subset \mathbb{R}^2$ that realizes $A$ with the minimum number of crossings.

\paragraph{\bf Angle sequences in 3-space and spherical polygonal linkages.}
In $\mathbb{R}^d$, $d\geq 3$, the sign of a turning angle no longer plays a role: The \emph{turning angle} of an oriented polygon $P$ at $v_i$ is in $(0,\pi)$, and an angle sequence $A=(\alpha_0,\ldots,\alpha_{n-1})$ is in $(0,\pi)^n$. The unit-length direction vectors of the edges of $P$ determine a spherical polygon $P'$ in $\mathbb{S}^{d-1}$. 
Note that the turning angles of $P$ correspond to the spherical lengths of the segments of $P'$. It is not hard to see that this observation reduces the problem of realizability of $A$ by a polygon in $\mathbb{R}^d$ to the problem of realizability of $A$  by a spherical polygon in $\mathbb{S}^{d-1}$, in the sense defined below, that additionally contains the origin $\mathbf{0}$ in the interior of its convex hull.

Let $\mathbb{S}^2\subset \mathbb{R}^3$ denote the unit 2-sphere.
A \emph{great circle} $C\subset \mathbb{S}^2$ is the intersection of $\mathbb{S}^2$
with a 2-dimensional hyperplane in $\mathbb{R}^3$ containing $\mathbf{0}$. A \emph{spherical line segment} is a connected subset of a great circle that does not contain a pair of antipodal points of $\mathbb{S}^2$.
The \emph{length} of a spherical line segment $ab$ equals the measure of the central angle subtended by $ab$. A \emph{spherical polygon} $P\subset \mathbb{S}^2$ is a closed  curve consisting of finitely many spherical segments; and a spherical polygon $P=(\mathbf{u}_0,\ldots, \mathbf{u}_{n-1})$,  $\mathbf{u}_i\in \mathbb{S}^2$, realizes an angle sequence $A=(\alpha_0,\ldots,\alpha_{n-1})$ if the spherical segment $(\mathbf{u}_{i-1},\mathbf{u}_{i})$ has (spherical) length $\alpha_i$, for $i=0,\ldots , n-1$. As usual, the \emph{turning angle} of $P$ 
at $\mathbf{u}_i$ is the angle in $[0,\pi]$ between  the tangents to $\mathbb{S}^2$ at  $\mathbf{u}_i$ that are co-planar with the great circles containing $(\mathbf{u}_i,\mathbf{u}_{i+1})$ and $(\mathbf{u}_i,\mathbf{u}_{i-1})$. Unlike for polygons in $\mathbb{R}^2$ and $\mathbb{R}^3$, we do not put any constraints on turning angles of spherical polygons 
(i.e., angles $0$ and $\pi$ are allowed).

Regarding realizations of $A$ by spherical polygons, we prove the following.

\begin{theorem}
\label{thm:3d_realization_char}
Let $A=(\alpha_0,\ldots,\alpha_{n-1})$, $n\ge 3$, be an angle sequence.
There exists a 
polygon $P\subset \mathbb{R}^3$ realizing $A$ if and only if $\sum_{i=0}^{n-1}\alpha_i\geq 2\pi$ and there exists a spherical polygon $P'\subset \mathbb{S}^2$ realizing $A$. Furthermore, $P$ can be constructed efficiently if $P'$ is given.
\end{theorem}


\begin{restatable}{theorem}{RealizationChar}
\label{thm:3d_realization_sphere}
There exists a constructive weakly polynomial-time algorithm to test whether  a given angle sequence $A=(\alpha_0,\ldots,\alpha_{n-1})$
can be realized by a spherical polygon $P'\subset \mathbb{S}^2$.
\end{restatable}

A simple exponential-time algorithm for realizability of angle sequences by spherical polygons follows from a known characterization~\cite[Theorem 2.5]{BS02_spherical}, which also implies that the order of angles in $A$ does not matter for the spherical realizability.
The topology of the configuration spaces of spherical polygonal linkages have also been studied~\cite{KM1999}. Independently, 
Streinu et al.~\cite{PaninaS10,StreinuW04} showed that the configuration space of \emph{noncrossing} spherical linkages is connected if $\sum_{i=0}^{n-1}\alpha_i\leq 2\pi$. However, these results do not seem to help prove Theorem~\ref{thm:3d_realization_sphere}.

The combination of Theorems~\ref{thm:3d_realization_char} and~\ref{thm:3d_realization_sphere} yields our second main result.

\begin{theorem}
\label{thm:3d_realization}
There exists a constructive weakly polynomial-time algorithm to test whether a given angle sequence $A=(\alpha_0,\ldots,\alpha_{n-1})$ can be realized by a polygon $P\subset \mathbb{R}^3$.
\end{theorem}

Our methods directly generalize from $\mathbb{R}^3$ to $\mathbb{R}^d$ for any integer $d\geq 3$. 
It turns out that higher dimensions do not translate to more realizable angle sequences. In particular, an angle sequence is realizable by a polygon in $\mathbb{R}^d$, $d\geq 3$, if and only if it is realizable in $\mathbb{R}^3$. 
We restrict ourselves to 2D and 3D in this paper. 

\paragraph{Organization.} We prove Theorem~\ref{thm:crossings_number} in Section~\ref{sec:2D} and
Theorems~\ref{thm:3d_realization_char},~\ref{thm:3d_realization_sphere},
and~\ref{thm:3d_realization} in Section~\ref{sec:3D}.
We show in Section~\ref{sec:3Dcross} that self-intersections are unavoidable in 3D if all realizations of an angle sequence are 2-dimensional.
We finish with concluding remarks in Section~\ref{sec:conclusion}.




\section{Crossing Minimization in the Plane}
\label{sec:2D}

The first part of the following lemma gives a folklore necessary condition for the consistency of an angle sequence $A$ in the plane. The condition is also sufficient except when $k=0$.
The second part  follows from a result of Gr\"unbaum and Shepard~\cite[Theorem~6]{grunbaum1990rotation}, using a decomposition due to  Wiener~\cite{Wiener64}. 
 We provide a proof for the sake of completeness.

\begin{figure}[ht]
\centering
\includegraphics[scale=1]{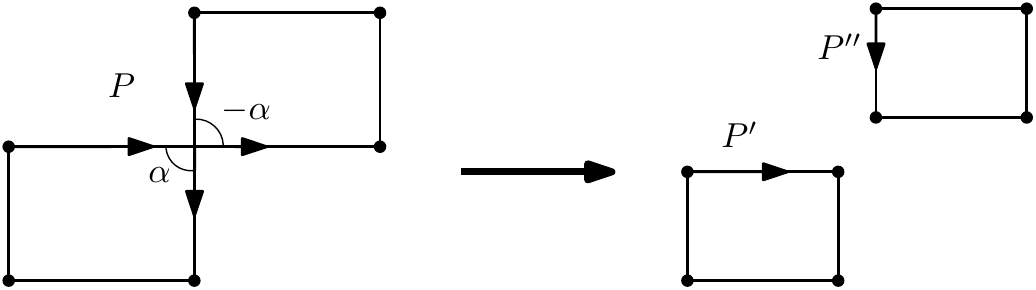}
\caption{Splitting an oriented closed polygon $P$ at a self-crossing point into 2 oriented closed polygons $P'$ and $P''$ such that $\mathrm{tn}(P)=\mathrm{tn}(P')+\mathrm{tn}(P'')$.}
\label{fig:splittingpolygon}
\end{figure}

\begin{lemma}
\label{lemma:0sum}
If an angle sequence $A=(\alpha_0,\ldots,\alpha_{n-1})$ is consistent, 
then $\sum_{i=0}^{n-1}\alpha_i=2k\pi$ for some $k\in \mathbb{Z}$,
and 
$\mathrm{cr}(A)\ge |k|-1$.
\end{lemma}

\begin{proof}
Since $A$ is consistent, $\sum_{i=0}^{n-1}\alpha_i=2k\pi$ for some $k\in \mathbb{Z}$, where $k=\mathrm{tn}(P)$ is the turning number of any generic realization $P$ of the angle sequence $A$. We prove by induction on $\mathrm{cr}(A)$ that $\mathrm{cr}(A)\ge |k|-1$.

In the base case, we have $\mathrm{cr}(A)=0$. 
Let $P$ be a generic realization of $A$ such that $\mathrm{cr}(P)=0$.
Then $P$ is a simple polygon with $n$ vertices. 
The internal angles of a simple $n$-gon sum up to $(n-2)\pi$. The internal angle of $P$ at vertex $v_i$ is $\pi-\alpha_i$ or $\pi+\alpha_i$, depending on the orientation of $P$. Thus $(n-2)\pi=\sum_{i=0}^{n-1}(\pi-\alpha_i)=(n-2k)\pi$ or $(n-2)\pi=\sum_{i=0}^{n-1}(\pi+\alpha_i)=(n+2k)\pi$.
Both cases yield $|\sum_{i=0}^{n-1}\alpha_i|=2\pi$, hence $|\mathrm{tn}(P)|=k=1$ and the claim follows.

In the inductive step, we have  $\mathrm{cr}(A)\ge 1$. 
Let $P$ be a generic realization of $A$ such that $\mathrm{cr}(A)= \mathrm{cr}(P)$; 
refer to Fig.~\ref{fig:splittingpolygon}.
By splitting $P$ at a self-crossing, we obtain a pair of closed polygons $P'$ and $P''$ such that $\mathrm{tn}(P)=\mathrm{tn}(P')+\mathrm{tn}(P'')$.
Since $\mathrm{cr}(P')<\mathrm{cr}(P)$ and $\mathrm{cr}(P'')<\mathrm{cr}(P)$,
induction yields $\mathrm{cr}(P) = 1 + \mathrm{cr}(P')+\mathrm{cr}(P'')\ge 1+ |\mathrm{tn}(P')|-1+ |\mathrm{tn}(P'')|-1\ge |\mathrm{tn}(P)|-1$, as required.
\end{proof}

The following lemma shows that the lower bound in Lemma~\ref{lemma:0sum} is tight when $\alpha_i>0$ for all $i\in \{0,\ldots , n-1\}$.

\begin{lemma}
\label{lemma:0sum2}
If $A=(\alpha_0,\ldots,\alpha_{n-1})$ is an angle sequence such that $\sum_{i=0}^{n-1}\alpha_i=2k\pi$, $k\neq 0$, and  $\alpha_i>0$ for all $i$, 
then  $\mathrm{cr}(A)\le |k|-1$.
\end{lemma}

\begin{figure}[t]
\centering
\includegraphics[scale=1]{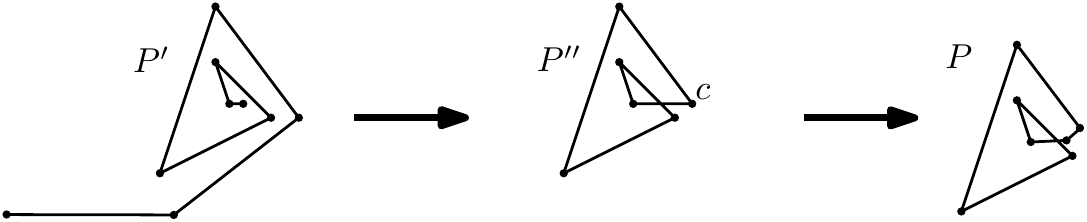}
\caption{Constructing a polygon $P$ with $|\mathrm{tn}(P)|-1$ crossings.}
\label{fig:spiral}
\end{figure}

\begin{proof}
Refer to Fig.~\ref{fig:spiral}.
In three steps, we construct a polygon $P$ realizing $A$ with $|\mathrm{tn}(P)|-1$ self-crossings thereby proving  $\mathrm{cr}(A)\le|k|-1=|\mathrm{tn}(P)|-1$. In the first step, we construct an oriented self-crossing-free  polygonal line $P'$ with $n+2$ vertices, whose first and last (directed) edges are parallel to the positive $x$-axis, and whose internal vertices have turning angles $\alpha_0,\ldots,\alpha_{n-1}$ in this order.
We construct $P'$ incrementally: The first edge has unit length starting from the origin; and every successive edge lies on a ray emanating from the endpoint of the previous edge. If the ray intersects neither the $x$-axis nor previous edges, then let the next edge have unit length, otherwise its length is chosen to avoid any such intersection. 

Let $S'$ be the last (directed) edge of $P'$, and let $\ell$ be the (horizontal) supporting line of $S'$. Since $\alpha_i>0$, for all $i$, the non-horizontal portions of $P'$ can be partitioned into $2k$ maximal $y$-monotone paths: $k$ increasing and $k$ decreasing paths. By construction, these paths are pairwise non-crossing, their $y$-extents, that is, the projections to the $y$-axis, are pairwise nested intervals, where each interval contains subsequent intervals. Consequently, $\ell$ intersects all $2k$ $y$-monotone paths. In particular, it crosses $k$ increasing paths to the right of $S'$, and meets all $k$ decreasing path at or to the left of $S'$. 

In the second step, extend $S'$ to the right until its  rightmost intersection point $c$
with a $y$-monotone increasing path of $P'$;
and denote by $P''$ the resulting closed polygon composed of the part of $P'$ from $c$ to $c$ via the extended segment $S'$. Note that $P''$ has $k-1$ self-intersections, as the extension of $S'$ crosses $P'$ in $k-1$ points.
Finally, we construct $P$ realizing $A$ from $P''$ by a modification of $P''$ in a small neighborhood of $c$ without creating additional self-crossings. 
Specifically, we replace the neighborhood of $c$ with a scaled 
copy of the initial portion of $P'$ between the first vertex of $P'$ and $c$.
\end{proof}

To prove the upper bound in Theorem~\ref{thm:crossings_number}, it remains to consider the case that $A=(\alpha_0,\ldots, \alpha_{n-1})$ contains both positive and negative angles. 
The crucial notion in the proof is that of an (essential) sign change of $A$ which we define next.
Let $\beta_i=\sum_{j=0}^{i}\alpha_j \mod 2\pi$ for $i=0,\ldots , n-1$.
Let $\mathbf{v}_i\in \mathbb{R}^2$ denote the unit vector $(\cos \beta_i, \sin \beta_i)$. Hence, $\mathbf{v}_i$ is the direction vector of the $(i+1)$-st edge of an oriented polygon $P$ realizing $A$ if the direction vector of the first edge of $P$ is $(1,0)\in \mathbb{R}^2$. 
By Garg's observation~\cite[Section~6]{G98_anglegraphs}, the consistency of $A$ implies that $\mathbf{0}$ is a strictly positive convex combination of vectors $\mathbf{v}_i$, that is, there exist scalars $\lambda_0,\ldots, \lambda_{n-1}> 0$ such that $\sum_{i=0}^{n-1}\lambda\mathbf{v}_i=\mathbf{0}$ and $\sum_{i=0}^{n-1}\lambda_i=1$.

The \emph{sign change} of $A$ is an index $i\in \{0,\ldots , n-1\}$ such that $\alpha_i\cdot \alpha_{i+1}<0$, where arithmetic on the indices is taken modulo $n$. Let $\mathrm{sc}(A)$ denote the number of sign changes of $A$. 
Note that the number of sign changes of $A$ is even.
A sign change $i\in \{0,\ldots , n-1\}$ of a consistent angle sequence $A$ is \emph{essential} if 
$\mathbf{0}$ is not a strictly positive convex combination of $\{\mathbf{v}_0,\ldots , \mathbf{v}_{i-1},\mathbf{v}_{i+1},\ldots , \mathbf{v}_{n-1}\}$.

\begin{lemma}
\label{lemma:all-essential}
If $A=(\alpha_0,\ldots,\alpha_{n-1})$ is a consistent angle sequence, where $\sum_{i=0}^{n-1}\alpha_i=2k\pi$, $k\in \mathbb{Z}$, 
and all sign changes are essential, 
then $\mathrm{cr}(A)\le \big||k|-1\big|$.
  \end{lemma}

\begin{proof}
We distinguish between two cases depending on whether $\sum_{i=0}^{n-1}\alpha_i=0$.

\noindent \textbf{Case 1: $\sum_{i=0}^{n-1}\alpha_i=0$.}
Since $\sum_{i=0}^{n-1}\alpha_i=0$, we have $\mathrm{sc}(A)\ge 2$.
Since all sign changes are essential, for any two distinct sign changes $i\neq j$, 
we have $\mathbf{v}_i\neq \mathbf{v}_j$, therefore counting different vectors ${\bf v}_i$, 
where $i$ is a sign change, is equivalent to counting essential sign changes.

We show next that $\mathrm{sc}(A)=2$. Suppose, to the contrary, that $\mathrm{sc}(A)>2$. 
Since the number of sign changes in a cyclic sequence of signs is even, we have $\mathrm{sc}(A)\ge 4$. We observe that if $\mathbf{v}_i$ corresponds to an essential sign change $i$, then there exists an open halfplane $H_i$ bounded by a line through the origin that such that 
$H_i\cap \{\mathbf{v}_0,\ldots ,\mathbf{v}_{n-1}\}=\{\mathbf{v}_i\}$. 
Let $i$, $j$, $i'$, and $j'$ be distinct essential sign changes such that 
$\mathbf{v}_i$, $\mathbf{v}_j$, $\mathbf{v}_{i'}$, and $\mathbf{v}_{i'}$ are in cyclic order around the origin. Since $H_i$ and $H_{i'}$ contains neither $\mathbf{v}_j$ nor $\mathbf{v}_{j'}$,
then $H_i$ and $H_{i'}$ are disjoint, lying on opposite sides of a line, which necessarily  contains both $\mathbf{v}_j$ and $\mathbf{v}_{j'}$. In particular, we have $\mathbf{v}_b=-\mathbf{v}_d$. Analogously, we can show that $\mathbf{v}_a=\mathbf{v}_d$.
Since $j$ is an sign change, either $H_i$ or $H_{i'}$ contains both $\mathbf{v}_{j-1}$ and $\mathbf{v}_{j+1}$. Thus there exists a fifth vector $\mathbf{v}_k$, which implies that one of $i$, $i'$, $j$, and $j'$ is not essential (contradiction).
 
Assume w.l.o.g.\ that the only two sign changes are $j$ and $n-1$, for some  $j\in \{0,\ldots ,n-2\}$.
We claim that $\mathbf{v}_{j}\neq -\mathbf{v}_{n-1}$. Suppose, to the contrary, that 
$\mathbf{v}_{j}=-\mathbf{v}_{n-1}$. Since both sign changes are essential, all vectors $\mathbf{v}_i$, other than $\mathbf{v}_{j}$ and $\mathbf{v}_{n-1}$, are outside of $H_j\cup H_{n-1}$. If $H_j\cap H_{n-1}\neq \emptyset$, then these vectors are an open half-plane bounded by the line through $\mathbf{v}_{j}$ and $-\mathbf{v}_{n-1}$. However, then $\mathbf{0}$ is not a strict convex combination of the vectors $\{\mathbf{v}_0,\ldots , \mathbf{v}_{n-1}$, contradicting the consistency of $A$. Hence we may assume that $H_j$ and $H_{n-1}$ are disjoint, and they lie on opposite sides of a line through the origin. Due to the consistency of $A$, there exists a pair $\{i,i'\}$ such that $\mathbf{v}_i=-\mathbf{v}_{i'}$. However, $j$ and $n-1$ are the only sign changes by assumption, and thus there exists a fifth index $\ell$ such that $\mathbf{v}_\ell\neq\pm \mathbf{v}_{i}$ (contradiction).

\begin{figure}[htb]
\centering
\includegraphics[scale=1]{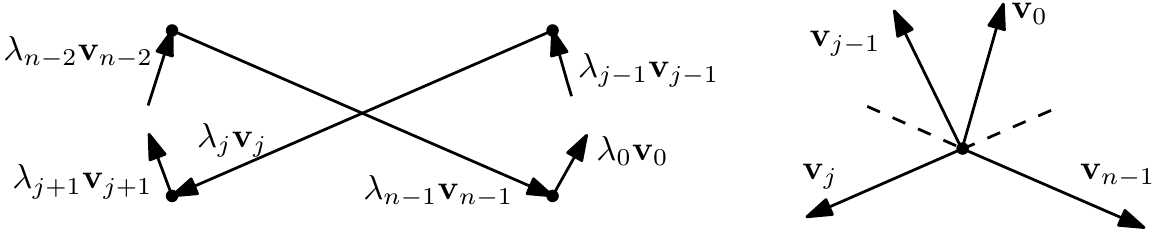}
\caption{The case of exactly 2 sign changes $j$ and $n-1$, both of which are essential, when $\sum_{i=0}^{n-1}\alpha_i=0$. Both missing parts of the polygon on the left are convex chains.}
\label{fig:2essential}
\end{figure}

We may assume that ${\bf v}_j$ and ${\bf v}_{n-1}$ are not collinear, and that the remaining vectors ${\bf v}_i$ belong to the closed convex cone bounded by $-{\bf v}_j$ and $-{\bf v}_{n-1}$; refer to Fig.~\ref{fig:2essential}.
 Thus, we may assume that (i)~$\beta_{n-1}=0$, (ii)~the sign changes of $A$ are $j$ and $n-1$, and (iii)~$0<\beta_0<\ldots <\beta_{j}$ and $\beta_{j}>\beta_{j+1}>\ldots >\beta_{n-1}=0$.
 Now, realizing $A$ by a generic
 polygon with exactly 1 crossing between the line segments in the direction of $\mathbf{v}_{j}$ and $\mathbf{v}_{n-1}$ is a simple exercise.

\noindent \textbf{Case 2: $\sum_{i=0}^{n-1}\alpha_i\neq 0$.}
We show that, unlike in the first case, none of the sign changes of $A$ can be essential. Indeed, suppose $j$ is an essential sign change, and let
$A'=(\alpha_0',\ldots, \alpha_{n-2}')=(\alpha_0,\ldots, \alpha_{j-1}, \alpha_{j}+\alpha_{j+1},\ldots, \alpha_{n-1})$ and $\beta_i'=\sum_{j=0}^{i}\alpha_j' \mod 2\pi$.
Consider the unit vectors $\mathbf{v}_0',\ldots ,\mathbf{v}_{n-2}'$, where 
 $\mathbf{v}_i'=(\cos \beta_i', \sin \beta_i')$. Since $j$ is an essential sign change, there exists a nonzero vector $\mathbf{v}$ such that $\big\langle \mathbf{v},\mathbf{v}_j\big\rangle> 0$ and $\big\langle \mathbf{v},\mathbf{v}_i'\big\rangle\le 0$ for all $i$. Hence, by symmetry, we may assume that $0\le\beta_i'\le \pi$, for all $i$. 
 Since $j$ is a sign change, we have $-\pi<\alpha_i'<\pi$ for all $i$, consequently $\beta_j'=\sum_{i=0}^{j}\alpha_i'\mod{2\pi}=\sum_{i=0}^{j}\alpha_i'$,
 which in turn implies, by Lemma~\ref{lemma:0sum}, that $0=\beta_{n-2}'= \sum_{i=0}^{n-2}\alpha_i'=\sum_{i=0}^{n-1}\alpha_i$ (contradiction). 
 
We have shown that $A$ has no sign changes. 
By Lemma~\ref{lemma:0sum2}, we have $\text{cr}(A)\le |k|-1$, which concludes the proof.
\end{proof}

\CrossingNumber*

\begin{proof} 
The claimed lower bound $\mathrm{cr}(A)\ge \big||k|-1\big|$ on the crossing number of $A$ follows by Lemma~\ref{lemma:0sum}, in the case when $k\neq 0$, and the result of Viyajan~\cite[Theorem~2]{vijayan1986geometry} in the case when $k=0$.
It remains to prove the upper bound $\mathrm{cr}(A)\le \big||k|-1\big|$.

We proceed by induction on $n$. In the base case, we have $n=3$. Then $P$ is a triangle,
$\sum_{i=0}^{2}\alpha_i =\pm 2\pi$, and $\mathrm{cr}(A)=0$, as required.
In the inductive step, assume $n\geq 4$, and that the claim holds for all shorter angle sequences. Let $A=(\alpha_0,\ldots, \alpha_{n-1})$ be an angle sequence with $\sum_{i=0}^{n-1}\alpha_i=2k\pi$.

If $A$ has no sign changes or if all sign changes are essential, then Lemma~\ref{lemma:0sum2} or Lemma~\ref{lemma:all-essential} completes the proof. Otherwise, there is at least one nonessential sign change. Let $s\in \{0,\ldots, n-1\}$ be a nonessential sign change and let $A'=(\alpha_0',\ldots, \alpha_{n-2}')=(\alpha_0,\ldots, \alpha_{s-1}, \alpha_{s}+\alpha_{s+1},\ldots, \alpha_{n-1})$. Note that $\sum_{i=0}^{n-2}\alpha_i'=2k\pi$. 
We eliminate $\alpha_s +\alpha_{s+1}$ from $A'$ if it is equal to 0.
Since the sign change $s$ is nonessential, $\mathbf{0}$ is a strictly positive convex combination of $\{\beta_0',\ldots , \beta_{n-2}'\}$, where $\beta_i'=\sum_{j=0}^{i}\alpha_j' \mod 2\pi$ for $i=0,\ldots, n-2$. Indeed, this follows from the fact that $\beta_i'=\beta_i$, for $i<s$, and $\beta_i'=\beta_{i+1}$, for $i\ge s$.

By the induction hypothesis, we obtain a realization of $A'$ as a generic polygon $P'$ with $\big||k|-1\big|$ crossings.
Let $v$ be a vertex of $P'$ corresponding to $\alpha_s+\alpha_{s+1}$, which is incident to sides $S_{s-1}'$ and $S_s'$ of $P'$  parallel to vectors ${\bf v}_{s-1}={\bf v}_{s-1}'$  and ${\bf v}_{s+1}={\bf v}_{s}'$.
We construct a generic polygon realizing $A$ by modifying $P$ in a small neighborhood of $v$ without introducing crossings, similarly to the method developed by Guibas et al.~\cite{GHS00_morph} as follows.
If $\alpha_s +\alpha_{s+1}=0$, then $\alpha_s+\alpha_{s+1}$ is eliminated from the sequence $A'$. We define $v$ as a vertex corresponding to $\alpha_{s+2}$ in this case.

\begin{figure}[htb]
\centering
\includegraphics[scale=1]{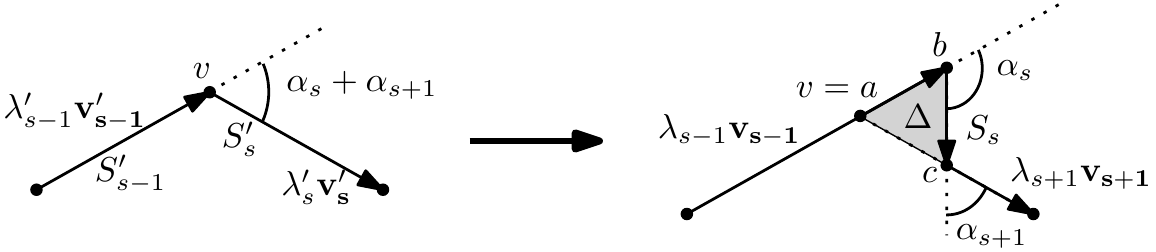}
\caption{Re-introducing the $s$-th vertex to the polygon $P'$ realizing $A'$ in order to obtain a polygon $P$ realizing $A$ when $\alpha_s+\alpha_{s+1}\neq 0$.}
\label{fig:3essential}
\end{figure}

First, we consider the case that $\alpha_s +\alpha_{s+1}\neq 0$.
Assume w.l.o.g.\ that $\alpha_s$ and $\alpha_s +\alpha_{s+1}$ have the same sign;
refer to Fig.~\ref{fig:3essential}.
Then there exists a triangle $\Delta=\Delta(abc)$ such that $\vec{ab}$, $\vec{bc}$, and $\vec{ca}$ are positive multiples of $\mathbf{v}_{s-1}=\mathbf{v}_{s-1}'$, $\mathbf{v}_s$, and $-\mathbf{v}_{s+1}=-\mathbf{v}_{s}'$, respectively. 
By a suitable translation, we may assume that $a=v$; and by a suitable scaling, we may assume that $\Delta$ is disjoint from all sides of $P'$ other than $S_{s-1}'$ and $S_s'$. Then we construct $P$ from $P'$ as follows. We extend $S_{s-1}'$ beyond $v=a$ with segment $ab$, insert a new side $bc$, and shorten $S_{s}'$ by removing segment $ac=vc$.
 
\begin{figure}[htb]
\centering
\includegraphics[scale=1]{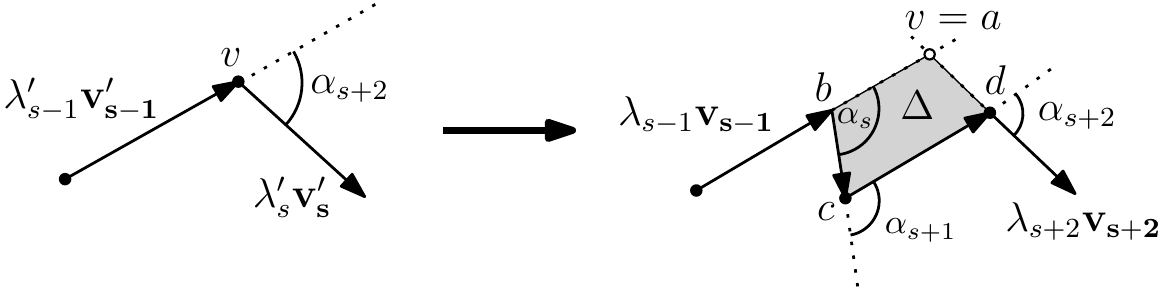}
\caption{Re-introducing the $s$-th and $(s+1)$-st vertex to the polygon $P'$ realizing $A'$ in order to obtain a polygon $P$ realizing $A$ when $\alpha_s+\alpha_{s+1}=0$.}
\label{fig:3essential2}
\end{figure}
 
It remains to consider the case that $\alpha_s +\alpha_{s+1}=0$.
Assume w.l.o.g.\ that $\alpha_s$ and $\alpha_{s+2}$ have the same sign;
refer to Fig.~\ref{fig:3essential2}.
Then there exists a trapezoid $\Delta=\Delta(abcd)$ such that $\vec{ab}$, $\vec{bc}$, $\vec{cd}$, and $\vec{da}$ are positive multiples of $-\mathbf{v}_{s-1}=-\mathbf{v}_{s-1}'$, $\mathbf{v}_s$, $\mathbf{v}_{s+1}=\mathbf{v}_{s-1}'$, and $-\mathbf{v}_{s+2}=-\mathbf{v}_{s}'$, respectively. 
By a suitable translation, we may assume that $a=v$; and by a suitable scaling, we may assume that $\Delta$ is disjoint from all sides of $P'$ other than $S_{s-1}'$ and $S_s'$. Then we construct $P$ from $P'$ as follows. We shorten $S_{s-1}'$ by removing segment $ab=vb$, insert two new sides $bc$ and $cd$, and shorten $S_{s}'$ by removing segment $da=dv$.
 \end{proof}

\section{Realizing Angle Sequences in 3-Space}
\label{sec:3D}

In this section, we describe a polynomial-time algorithm to decide whether an angle sequence $A=(\alpha_0,\ldots \alpha_{n-1})\in (0,\pi)^n$ can be realized as a polygon in $\mathbb{R}^3$.

We note that this problem is equivalent to solving a system of polynomial equations, where $3n$ variables describe the coordinates of the $n$ vertices of $P$, and each of $n$ equations is obtained by the cosine theorem applied for a vertex and two incident edges of $P$. However, it is unclear how to solve such a system efficiently.

By Fenchel's theorem in differential geometry~\cite{F51_curve}, the total curvature of a smooth curve in $\mathbb{R}^d$ is at least $2\pi$, and the curves with the total curvature equal to $2\pi$ must be plane. Fenchel's theorem has been adapted to closed polygons~\cite[Theorem~2.4]{Sullivan2008}, and it gives the following a necessary condition for an angle sequence $A$ to have a realization in $\mathbb{R}^d$, for all $d\geq 2$:
\begin{equation}
\label{eqn:3d_necessary}
\sum_{i=0}^{n-1}\alpha_i \ge 2\pi,
\end{equation}
and if $\sum_{i=0}^{n-1}\alpha_i \ge 2\pi$, then any realization lies in a plane (an affine subspace of $\mathbb{R}^d$). 
We show that a slightly stronger condition is both necessary and sufficient, hence it characterizes realizable angle sequences in $\mathbb{R}^3$.

\begin{lemma}\label{lemma:3d_realization_char}
Let $A=(\alpha_0,\ldots,\alpha_{n-1})$, $n\ge 3$, be an angle sequence.
There exists a polygon $P\subset \mathbb{R}^3$ realizing $A$ if and only if
there exists a spherical polygon $P'\subset \mathbb{S}^2$ realizing $A$ such that $\mathbf{0}\in \mathrm{relint}(\mathrm{conv}(P'))$ (relative interior of $\mathrm{conv}(P')$). 
Furthermore, $P$ can be constructed efficiently if $P'$ is given.
\end{lemma}

\begin{proof}
Assume that an oriented polygon $P=(v_0,\ldots ,v_{n-1})$ realizes $A$ in $\mathbb{R}^3$. 
Let $\mathbf{u}_i=(v_{i+1}-v_i)/\|v_{i+1}-v_i\|\in \mathbb{S}^2$ be the unit direction vector of the edge $v_i v_{i+1}$ of $P$ according to its orientation. Then $P'=(\mathbf{u}_0,\ldots ,\mathbf{u}_{n-1})$ is a spherical polygon that realizes $A$. Suppose, for the sake of contradiction, that $\mathbf{0}$ is not in the relative interior of $\mathrm{conv}(P')$.
Then there is a plane $H$ that separates $\mathbf{0}$ and $P'$, that is, if $\mathbf{n}$ is the normal vector of $H$, then $\big\langle \mathbf{n},\mathbf{u}_i\big\rangle>0$ for all $i\in \{0,\ldots , n-1\}$. This implies $\big\langle  \mathbf{n},(v_{i+1}-v_i)\big\rangle>0$ for all $i$, hence $\big\langle  \mathbf{n},\sum_{i=1}^{n-1}(v_{i+1}-v_i)\big\rangle>0$, which contradicts the fact that $\sum_{i=1}^{n-1}(v_{i+1}-v_i)=\mathbf{0}$, and $\big\langle  \mathbf{n},\mathbf{0}\big\rangle=0$. 

Conversely, assume that a spherical polygon $P'$ realizes $A$, with edge lengths $\alpha_0,\ldots, \alpha_{n-1}>0$. If all the vertices of $P'$ lie on a common great circle, then $\mathbf{0}\in \mathrm{relint}(\mathrm{conv}(P'))$ implies $\sum_{i=0}^{n-1}\pm\alpha_i= 0 \mod 2\pi$, where the sign is determined by the direction (cw. or ccw.) in which a particular segment of $P'$ traverses the common great circle according to its orientation.
As observed by Garg~\cite[Section 6]{G98_anglegraphs}, the signed angle sequence is consistent in this case due to the assumption that $\mathbf{0}\in \mathrm{relint}(\mathrm{conv}(P'))$. Thus, we obtain a realization of $A$ that is contained in a plane.

Otherwise we may assume that $\mathbf{0}\in \mathrm{int}(\mathrm{conv}(P'))$.
By Carath\'{e}odory's theorem~\cite[Thereom 1.2.3]{matousek2013lectures}, $P'$ has 4 vertices whose convex combination is the origin $\mathbf{0}$. Then we can express $\mathbf{0}$ as a strictly positive convex combination of \emph{all} vertices of $P'$. The coefficients in the convex combination encode the lengths of the edges of a polygon $P$ realizing $A$, which concludes the proof in this case.

We now show how to compute strictly positive coefficients in strongly polynomial time. Let $\mathbf{c}=\frac{1}{n}\sum_{i=0}^{n-1}\mathbf{u}_i$ be the centroid of the vertices of $P'$. 
If $\mathbf{c}=\mathbf{0}$, we are done. Otherwise, we can find a tetrahedron $T=\mathrm{conv}\{\mathbf{u}_{i_0},\ldots, \mathbf{u}_{i_3}\}$ such that $\mathbf{0}\in T$ 
and such that the ray from $\mathbf{0}$ in the direction $-\mathbf{c}$  intersects $\mathrm{int}(T)$, by solving an LP feasibility problem in $\mathbb{R}^3$. 
By computing the intersection of the ray with the faces of $T$, we find the maximum $\mu>0$ such that $-\mu\mathbf{c}\in \partial T$ (the boundary of $T$). We have $-\mu\mathbf{c}=\sum_{j=0}^3 \lambda_j \mathbf{u}_{i_j}$ and $\sum_{j=0}^3\lambda_j=1$ for suitable coefficients $\lambda_j\geq 0$. Now $\mathbf{0}=\mu\mathbf{c}-\mu\mathbf{c}=\frac{\mu}{n}\sum_{i=0}^{n-1}\mathbf{u}_i+\sum_{j=0}^3 \lambda_j\mathbf{u}_{i_j}$ is a strictly positive convex combination of the vertices of $P'$. 
\end{proof}

It is easy to find an angle sequence $A$ that satisfies~\eqref{eqn:3d_necessary} but does not correspond to a spherical polygon $P'$. 
Consider, for example,  $A=(\pi-\varepsilon,\pi-\varepsilon,\pi-\varepsilon,\varepsilon)$, 
for some small $\varepsilon>0$. 
Points in $\mathbb{S}^2$ at (spherical) distance $\pi-\varepsilon$ are nearly antipodal. 
Hence, the endpoints of a polygonal chain $(\pi-\varepsilon,\pi-\varepsilon,\pi-\varepsilon)$ are nearly antipodal as well, and cannot be connected by an edge of (spherical) length $\varepsilon$. Thus a spherical polygon cannot realize $A$.

\paragraph{Algorithms.} 
In the remainder of this section, we show how to find a realization $P\subset \mathbb{R}^3$ or report that none exists, in polynomial time. Our first concern is to decide whether an angle sequence is realizable by a spherical polygon. This is possible to do in a weakly polynomial-time.

\RealizationChar*
  
\begin{proof}
Let $A=(\alpha_0,\ldots , \alpha_{n-1})\in (0,\pi)^n$ be a given angle sequence.
Let $\mathbf{n}=(0,0,1)\in \mathbb{S}^2$, that is, $\mathbf{n}$ is the north pole. 
For $i\in \{0,1,\ldots , n-1\}$,
let $U_i\subseteq \mathbb{S}^2$ be the locus of the end vertices $\mathbf{u}_i$ of all (spherical) polygonal lines $P_i'=(\mathbf{n},\mathbf{u}_0,\ldots ,\mathbf{u}_i)$ with edge lengths $\alpha_0,\ldots , \alpha_{i-1}$. It is clear that $A$ is realizable by a spherical polygon $P'$ if and only if $\mathbf{n}\in U_{n-1}$. 

Note that for all $i\in \{0,\ldots , n-1\}$, the set $U_i$ is invariant under rotations about the $z$-axis, since $\mathbf{n}$ is a fixed point and rotations are isometries. We show how to compute the sets $U_i$, $i\in \{0,\ldots ,n-1\}$, efficiently. 

We define a \emph{spherical zone} as a subset of $\mathbb{S}^2$ between two horizontal planes (possibly, a circle, a spherical cap, or a pole). Recall the parameterization of $\mathbb{S}^2$ using spherical coordinates (cf. Figure~\ref{fig:param} (left)):  for every $\mathbf{v}\in \mathbb{S}^2$, $\mathbf{v}(\psi, \varphi)=(\sin \psi\sin\varphi,\cos\psi\sin\varphi, \cos \varphi)$, with longitude $\psi\in [0,2\pi)$ and polar angle $\varphi\in [0,\pi]$, where the \emph{polar angle} $\varphi$ is the angle between $\mathbf{v}$ and $\mathbf{n}$. 
Using this parameterization, a spherical zone is a Cartesian product $[0,2\pi)\times I$ for some circular arc $I\subset [0,\pi]$. In the remainder of the proof, we associate each spherical zone with such a circular arc $I$.

We define additions and subtraction on polar angles $\alpha,\beta\in [0,\pi]$ by 
%
\[
\alpha\oplus\beta=\min\{\alpha+\beta,2\pi-(\alpha+\beta)\}, \,\,
\alpha\ominus\beta=\max\{\alpha-\beta,\beta-\alpha\};
\]
 see Figure~\ref{fig:param} (right). (This may be interpreted as addition mod $2\pi$, restricted to the quotient space defined by the equivalence relation $\varphi\sim 2\pi-\varphi$.)
\begin{figure}[htb]
\centering
\includegraphics[scale=0.7]{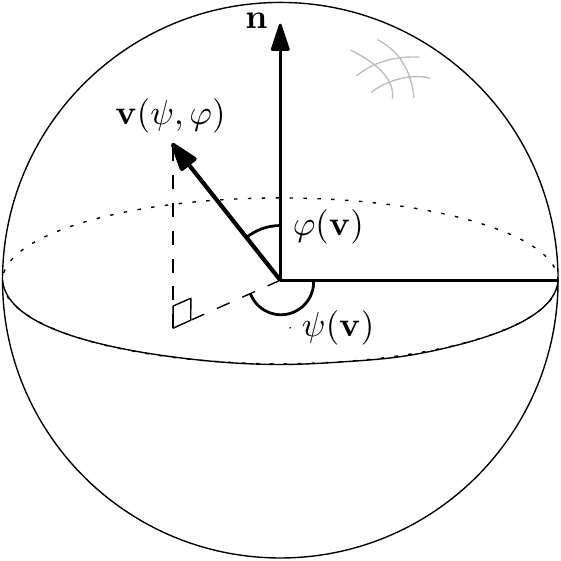} \hspace{30pt}
\includegraphics[scale=0.65]{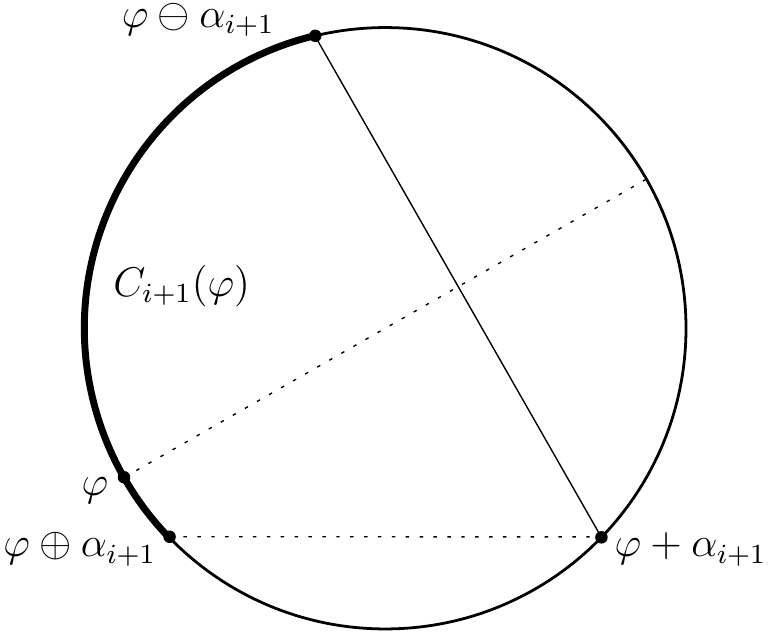}
\caption{Parametrization of the unit vectors (left). Circular arc $C_{i+1}(\varphi)$ (right).}
\label{fig:param}
\end{figure}

We show that $U_i$ is a spherical zone for all $i\in \{0,\ldots , n-1\}$, and show how to compute the intervals $I_i\subset [0,\pi]$ efficiently. First note that $U_0$ is a circle at (spherical) distance $\alpha_0$ from $\mathbf{n}$, hence $U_0$ is a spherical zone with $I_0=[\alpha_0,\alpha_0]$. 

Assume that $U_i$ is a spherical zone associated with $I_i\subset [0,\pi]$. 
Let $\mathbf{u}_i\in U_i$, where $\mathbf{u}_i=\mathbf{v}(\psi,\varphi)$ with $\psi\in [0,2\pi)$ and $\varphi\in I_i$. By the definition $U_i$, there exists a polygonal line $(\mathbf{n},\mathbf{u}_0,\ldots , \mathbf{u}_i)$ with edge lengths $\alpha_0,\ldots , \alpha_i$. The locus of points in $\mathbb{S}^2$ at distance $\alpha_{i+1}$ from $u_i$ is a circle; the polar angles of the points in the circle form an interval $C_{i+1}(\varphi)$. 
Specifically (see Figure~\ref{fig:param} (right)), we have
%
\[
C_{i+1}(\varphi)=[\min\{\varphi\ominus\alpha_{i+1},\varphi\oplus\alpha_{i+1} \},\max\{\varphi\ominus\alpha_{i+1},\varphi\oplus\alpha_{i+1}\}].
\]
By rotational symmetry, $U_{i+1}=[0,2\pi)\times I_{i+1}$, where 
$I_{i+1}= \bigcup_{\varphi\in I_i} C_{i+1}(\varphi)$. Consequently, $I_{i+1}\subset [0,\pi]$ is  connected, and hence, $I_{i+1}$  is an interval. Therefore  $U_{i+1}$ is a spherical zone. 
As $\varphi\oplus \alpha_{i+1}$ and $\varphi\ominus \alpha_{i+1}$ are 
piecewise linear functions of $\varphi$, we can compute $I_{i+1}$ using $O(1)$  arithmetic operations. 
 

We can construct the intervals $I_0,\ldots , I_{n-1}\subset [0,\pi]$ as described above.
If $0\notin I_{n-1}$, then $\mathbf{n}\notin U_{n-1}$ and $A$ is not realizable. Otherwise, we can compute the vertices of a spherical realization $P'\subset \mathbb{S}^2$ by backtracking. Put $\mathbf{u}_{n-1}=\mathbf{n}=(0,0,1)$. 
Given $\mathbf{u}_i=\mathbf{v}(\psi,\varphi)$, we choose $\mathbf{u}_{i-1}$ as follows. 
Let $\mathbf{u}_{i-1}$ be $\mathbf{v}(\psi,\varphi\oplus \alpha_i)$ or $\mathbf{v}(\psi,\varphi\ominus\alpha_i)$ if either of them is in $U_{i-1}$ (break ties arbitrarily). Else the spherical circle  of radius $\alpha_{i}$ centered at $\mathbf{u}_i$ intersects the boundary of $U_{i-1}$, and then we choose $\mathbf{u}_{i-1}$ to be an arbitrary such intersection point. The decision algorithm (whether $0\in I_{n-1}$) and the backtracking both use $O(n)$ arithmetic operations. 
%
\end{proof}

 \paragraph{Enclosing the Origin.} 
Theorem~\ref{thm:3d_realization_sphere} provides an efficient algorithm to test whether an angle sequence can be realized by a spherical polygon, however, Lemma~\ref{lemma:3d_realization_char} requires a spherical polygon $P'$ whose convex hull contains the origin in its relative interior. We show that this is always possible if a realization exists and $\sum_{i=0}^{n-1}\alpha_i\geq 2\pi$.
  The general strategy in the inductive proof of this claim (Lemma~\ref{lemma:convex-hull} below) is to incrementally modify $P'$ by changing the turning angle at one of its vertices to $0$ or $\pi$. This allows us to reduce the number of vertices of $P'$ and apply induction.

Before we are ready to prove Lemma~\ref{lemma:convex-hull}  we need to do some preliminary work. First, we introduce some terminology for spherical polygonal linkages with one fixed endpoint.
 Let $P'=(\mathbf{u}_0,\ldots, \mathbf{u}_{n-1})$ be a polygon in $\mathbb{S}^2$ that realizes an angle sequence $A=(\alpha_0,\ldots, \alpha_{n-1})$; 
 we do not assume $\sum_{i=0}^{n-1}\alpha_i\geq 2\pi$. 
Denote by $U_i^{j-}$ the locus of the endpoints $\mathbf{u}_i'\in \mathbb{S}^2$ of all (spherical) polygonal lines $(\mathbf{u}_{i-j},\mathbf{u}_{i-j+1}',\ldots ,\mathbf{u}_i')$, where the first vertex is fixed at $\mathbf{u}_{i-j}$, and the edge lengths are $\alpha_{i-j},\ldots , \alpha_{i}$.  Similarly, denote by $U_i^{j+}$ the locus of the endpoints $\mathbf{u}_i'\in \mathbb{S}^2$ of all (spherical) polygonal lines $(\mathbf{u}_{i+j},\mathbf{u}_{i+j-1}',\ldots ,\mathbf{u}_i')$ with edge lengths $\alpha_{i+j+1},\ldots , \alpha_{i+1}$. Due to rotational symmetry about the line passing through $\mathbf{u}_{i-j}$ and $\mathbf{0}$, the sets $U_i^{j-}$ and $U_i^{j+}$ are each a \emph{spherical zone} (i.e., a subset of $\mathbb{S}^2$ bounded by two parallel circles), possibly just a circle, or a cap, or a point.  
 In particular, the distance between $\mathbf{u}_i$ and any boundary component (circle) of $U_i^{j-}$ or $U_i^{j+}$ is the same; see Fig.~\ref{fig:strip}.

 \begin{figure}[htb]
 \centering
 \includegraphics[scale=1]{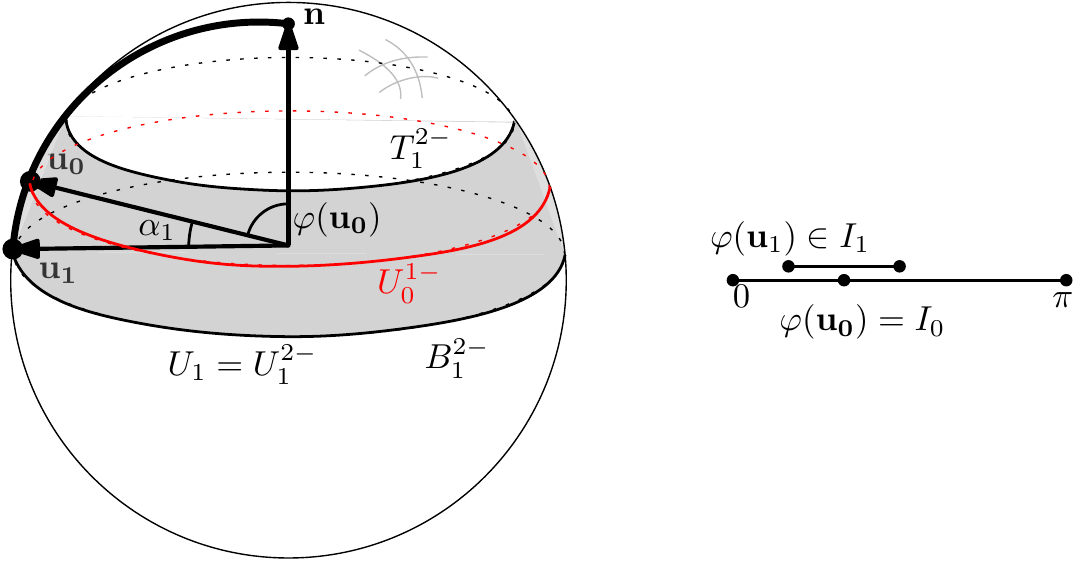}
 \caption{The spherical zone $U_1$ (or $U_1^{2-}$) containing ${\bf u_1}$ corresponding to $I_1$.}
 \label{fig:strip}
 \end{figure}

 If $U_i^{2+}$ is bounded by two circles, let $T_i^{2+}$ and $B_i^{2+}$ denote the two boundary circles such that $\mathbf{u}_i$ is closer to $T_i^{2+}$ than to $B_i^{2+}$. 
 If $U_i^{2+}$ is a cap, let $T_i^{2+}$ denote the boundary of $U_i^{2+}$, and let  $B_i^{2+}$ denote the center of $U_i^{2+}$. We define $T_i^{2-}$ and $B_i^{2-}$ analogously.

 The vertex $\mathbf{u}_i$ of $P'$ is a \emph{spur} of $P'$ if the segments $\mathbf{u}_i\mathbf{u}_{i+1}$ and $\mathbf{u}_i\mathbf{u}_{i-1}$ overlap
 (equivalently, the turning angle of $P'$ at $\mathbf{u}_i$ is $\pi$).
 We use the following simple but crucial observation. 
 \begin{observation}\label{obs:perturb}
 Assume that $n\ge 4$ and $U_i^{2+}$ is neither a circle nor a point.  
 The turning angle of $P'$ at $\mathbf{u}_{i+1}$ is 0 iff  $\mathbf{u}_{i}\in  B_i^{2+}$; 
 and $\mathbf{u}_{i+1}$ is a spur of $P'$ iff $\mathbf{u}_i\in T_i^{2+}$. (By symmetry, the same holds if we replace $+$ with $-$.)
 \end{observation}

 \begin{lemma}
 \label{lemma:straigthening}
 Let $P'$ be a spherical polygon $({\bf u}_0,\ldots, {\bf u}_{n-1})$, $n\ge 4$, that realizes an angle sequence $A=(\alpha_0,\ldots , \alpha_{n-1})$.
 Then there exists a spherical polygon $P''=({\bf u}_0,\ldots, {\bf u}_{i-1},{\bf u}_i',{\bf u}_{i+1}',{\bf u}_{i+2},\ldots, {\bf u}_{n-1})$ that also realizes $A$ such that 
 the turning angle at ${\bf u}_{i-1}$ is $0$, or the turning angle at  ${\bf u}_{i+1}'$ is $0$ or $\pi$.
  \end{lemma}

 \begin{proof}
 If $n\ge 4$, Observation~\ref{obs:perturb} allows us to move vertices $\mathbf{u}_i$ and $\mathbf{u}_{i+1}$ so that the turning angle at ${\bf u}_{i-1}$ drops to 0, or the turning angle at ${\bf u}_{i+1}$ changes to 0 or $\pi$, while all other vertices of $P'$ remain fixed. 
 Indeed, one of the following three options holds:  $U_i^{1-}\subseteq U_i^{2+}$,  $U_i^{1-}\cap B_i^{2+}\not=\emptyset$, or $U_i^{1-}\cap T_i^{2+}\not=\emptyset$. 
 If $U_i^{1-}\subseteq U_i^{2+}$, 
 then by Observation~\ref{obs:perturb} there exists ${\bf u}_i'\in U_i^{1-}\cap B_i^{2-}\cap  U_i^{2+}$. Since  ${\bf u}_i'\in  U_i^{2+}$ there exists ${\bf u}_{i+1}'\in   U_{i+1}^{1+}$
 such that $P''=({\bf u}_0,\ldots, {\bf u}_{i-1},{\bf u}_i',{\bf u}_{i+1}',{\bf u}_{i+2},\ldots, {\bf u}_{n-1})$ realizes $A$ and the turning angle at ${\bf u}_{i-1}$ equals 0.
 Similarly, if there exists ${\bf u}_i'\in U_i^{1-}\cap B_i^{2+}$ or ${\bf u}_i'\in U_i^{1-}\cap T_i^{2+}$, then there exists ${\bf u}_{i+1}'\in U_{i+1}^{1+}$ such that  $P''$ as above realizes $A$ with the turning angle at ${\bf u}_{i+1}$ equal to 0 or $\pi$, respectively.
 \end{proof}

 We are now ready to prove the lemma stated below.

\begin{lemma}\label{lemma:convex-hull}
Given a spherical polygon $P'$ that realizes an angle sequence $A=(\alpha_0,\ldots ,\alpha_{n-1})$, $n\ge 3$,  with $\sum_{i=0}^{n-1}\alpha\geq 2\pi$, we can compute in polynomial time a spherical polygon $P''$ realizing $A$ such that $\mathbf{0}\in \mathrm{relint}(\mathrm{conv}(P''))$.
\end{lemma}

 \begin{proof}
 We proceed by induction on the number of vertices of $P'$.
 In the basis step, we have $n=3$. In this case, $P'$ is a spherical triangle.
 The length of every spherical triangle is at most $2\pi$, contradicting the assumption that $\sum_{i=0}^{n-1}\alpha_i>2\pi$. Hence the claim vacuously holds. 

 In the induction step, assume that $n\geq 4$ and the claim holds for smaller values of $n$. 
 Assume $\mathbf{0}\notin \mathrm{relint}(\mathrm{conv}(P'))$, otherwise the proof is  complete. We distinguish between several cases. 

 \medskip\noindent \textbf{Case~1: a path of consecutive edges  lying in a great circle
  contains a half-circle.} We may assume w.l.o.g. that at least one endpoint of the half-circle is a vertex of $P'$. Since the length of each edge is less than $\pi$, the path that contains a half-circle has at least 2 edges.  

\medskip\noindent \textbf{Case~1.1: both endpoints of the half-circle are vertices of $P'$}. 
 Assume w.l.o.g., that the two endpoints of the half-circle are $\mathbf{u}_i$ and $\mathbf{u}_j$, for some $i<j$. These vertices decompose $P'$ into two polylines, $P'_1$ and $P'_2$. We rotate $P'_2$ about the line through $\mathbf{u}_i\mathbf{u}_j$ so that the turning angle at $\mathbf{u}_i$ is a suitable value in $[-\varepsilon,+\varepsilon]$ as follows. First, set the turning angle at $\mathbf{u}_i$ to be $0$. Let $P''$  denote the   resulting polygon.
  If  $\mathbf{0}\in \mathrm{int}(\mathrm{conv}(P''))$ we are done.
If  $P''$ is contained in a great circle then $\mathbf{0}\in \mathrm{int}(\mathrm{conv}(P''))$ due to the  angle $0$ at ${\bf u}_i$, and we are done as well. Else, $P''$ is contained in a hemisphere $H$ bounded by the great circle through $\mathbf{u}_{i-1}\mathbf{u}_i\mathbf{u}_{i+1}$. In this case, we perturb the turning angle at $\mathbf{u}_i$ so that $\mathbf{u}_{i+1}$ is not contained in $H$ thereby achieving $\mathbf{0}\in \mathrm{int}(\mathrm{conv}(P''))$.
 
 \medskip\noindent \textbf{Case~1.2: only one endpoint of the half-circle is a vertex of $P'$.} 
 Let $P'_1=(\mathbf{u}_i,\ldots, \mathbf{u}_j)$ be the longest path in $P'$ that contains a half-circle, and lies in a great circle. 
 Since $\mathbf{0}\notin \mathrm{relint}(\mathrm{conv}(P'))$, the polygon $P'$ is contained in a hemisphere $H$ bounded by the great circle $\partial H$ that contains $P'_1$, but $P'$ is not contained in $\partial H$. By construction of $P_1'$, we have $\mathbf{u}_{j+1}\notin\partial H$.
 In order to make the proof in this case easier, we make the following assumption.
 If a part $P_0$ of $P'$ between two antipodal/identical vertices that belong $\partial H$ is contained in a great circle, w.l.o.g.\ we assume that $P_0$  is contained in $\partial H$. (This can be achieved by a suitable rotation about the line passing through the endpoints of $P_0$.)

Assume, w.l.o.g.\ that the second endpoint of $P_1'$ is ${\bf u}_0$, that is, $j=0$. Let $j'$ be the smallest value such that ${\bf u}_{j'}\in \partial  H$. Since $\mathbf{0}\notin \mathrm{relint}(\mathrm{conv}(P'))$, we have ${\bf u}_{0},\ldots, {\bf u}_{j'}\in  H$. We show that we can perturb the polygon $P'$ into a new polygon $P''=({\bf u}_{0}',\ldots, {\bf u}_{j'-1}', {\bf u}_{j'},\ldots, {\bf u}_{n-1})$ realizing $A$ so that $\mathbf{0}\in \mathrm{int}(\mathrm{conv}(P''))$. Since   $({\bf u}_{0},\ldots, {\bf u}_{j'})$
   is not contained in a great circle by our assumption, there exists $j'', \ 0<j''<j'$, such that the turning angle of $P_1'$ at $j''$ is neither 0 nor $\pi$. We prove in the next paragraph that we can assume that $j''=1$.
   
   Suppose that $j''>1$. We perturb the polygon $P'$ thereby lowering its  value $j''$, while still keeping $P'$ a realization of $A$.
   By Observation~\ref{obs:perturb}, $\mathbf{u}_{j''-1}\notin \partial U_{j''-1}^{2+}$. Since the turning angle at $\mathbf{u}_{j''-1}$ is either 0 or $\pi$.  Note that $U_{j''-1}^{2+}$ is the union  of the spherical circles $S_{\bf c}$ of radius $\alpha_{j''-1}$ with centers ${\bf c}$ on $U_{j''}^{1+}$.
   Since $\mathbf{u}_{j''-1}\notin \partial U_{j''-1}^{2+}$, there exists a circle $S_{\bf c}$ that intersects $U_{j''-1}^{1-}$ in two different points ${\bf p}_1$ and ${\bf p}_2$. We replace ${\bf u}_{j''}$
   with ${\bf c}$ and ${\bf u}_{j''-1}$ with ${\bf p}_1$ on $P'$ thereby  still keeping $P'$ a realization of $A$. In the modified polygon $P'$, the turning angle at ${\bf u}_{j''-1}={\bf p}_1$ is neither 0 nor $\pi$.
   
   By~Observation~\ref{obs:perturb} and the assumption $j''=1$, we have $\mathbf{u}_{0}\notin \partial U_{0}^{2+}$, and we can perturb  ${\bf u}_0$ within $ U_{0}^{2+}$ into  ${\bf u}_0'$  and  ${\bf u}_1$ into ${\bf u}_1'$ so that ${\bf u}_{0}'\notin H$, and ${\bf u}_{1}',{\bf u}_2\ldots,{\bf u}_{j'-1}\in \mathrm{relint}(H)$, thereby achieving  $\mathbf{0}\in \mathrm{int}(\mathrm{conv}(P''))$.

 \medskip\noindent \textbf{Case~2: the turning angle of $P'$ is 0 at some vertex $\mathbf{u}_i$.}
 By supressing the vertex $\mathbf{u}_i$, we obtain a spherical polygon $Q'$ on $n-1$ vertices that realizes the sequence 
 $(\alpha_0,\ldots, \alpha_{i-2},\alpha_{i-1}+\alpha_{i},\alpha_{i+1},\ldots, \alpha_{n-1})$ unless
 $\alpha_{i-1}+\alpha_{i}\ge \pi$, but then we are in Case 1.
 By induction, this sequence has a realization $Q''$ such that $\mathbf{0}\in \mathrm{relint}(\mathrm{conv}(Q''))$. Subdivision of the edge of length $\alpha_{i-1}+\alpha_{i}$ producers a realization $P''$ of $A$ 
 such that $\mathbf{0}\in \mathrm{relint}(\mathrm{conv}(Q''))= \mathrm{relint}(\mathrm{conv}(P''))$.

 \medskip\noindent\textbf{Case~3: there is no path of consecutive edges lying in a great circle and  containing a half-circle, and no turning angle is 0.}

 \medskip\noindent\textbf{Case~3.1: $n= 4$.}  We claim that $U_0^{2+} \cap U_0^{2-}$ contains
 $B_0^{2-}$ or $B_0^{2+}$. By Observation~\ref{obs:perturb}, this immediately implies that 
 we can change one turning angle to 0 and proceed to Case~1. 

 To prove the claim, note that 
 $U_0^{2+} \cap U_0^{2-}\neq \emptyset$ and $-2 \equiv 2 \pmod 4$,  and hence the circles $T_0^{2-}$, $T_0^{2+}$, $B_0^{2-}$, and $B_0^{2+}$ are all parallel since they are all orthogonal to ${\bf u}_{2}$. Thus, by symmetry there are two cases to consider depending on whether $U_0^{2+} \subseteq U_0^{2-}$.
  If $U_0^{2+} \subseteq U_0^{2-}$, then $B_0^{2+}\subset U_0^{2+} \cap U_0^{2-}$.
  Else  $U_0^{2+} \cap U_0^{2-}$ contains $B_0^{2+}$ or $B_0^{2-}$, whichever is closer to ${\bf u}_{2}$, which concludes the proof of this case.

 \medskip\noindent\textbf{Case~3.2: $n\ge 5$.} 
  Choose $i\in\{0,\ldots, n-1\}$ so that $\alpha_{i+2}$ is a minimum angle in $A$.
 Note that $U_i^{2+}$ is neither a circle nor a point since that would mean that $\mathbf{u}_{i+2}$ and $\mathbf{u}_{i+1}$, or $\mathbf{u}_{i}$ and $\mathbf{u}_{i+1}$ are  antipodal, which is impossible.
 We apply Lemma~\ref{lemma:straigthening} and obtain a spherical polygon   
\[
P''=({\bf u}_0,\ldots, {\bf u}_{i-1},{\bf u}_i',{\bf u}_{i+1}',{\bf u}_{i+2},\ldots, {\bf u}_{n-1}).
\]
  If the turning angle of $P''$ at  ${\bf u}_{i-1}$ or  ${\bf u}_{i+1}'$ equals to 0, we proceed to Case~2. Otherwise, the turning angle of $P''$ at ${\bf u}_{i+1}'$ equals $\pi$. In other words, we introduce a spur at ${\bf u}_{i+1}'$.
  If $\alpha_{i+1}=\alpha_{i+2}$ we can make the turning angle of $P''$ at ${\bf u}_{i+2}$ equal to $0$ by rotating the overlapping segments $({\bf u}_{i+1}',{\bf u}_{i+2})$ and $({\bf u}_{i+1}',{\bf u}_{i}')$ around ${\bf u}_{i+2}={\bf u}_i'$ and proceed to Case~2.
 Otherwise, we have $\alpha_{i+2}<\alpha_{i+1}$ by the choice of $i$. Let $Q'$ denote an auxiliary polygon realizing  $(\alpha_0,\ldots, \alpha_{i},\alpha_{i+1}-\alpha_{i+2},\alpha_{i+3},\ldots, \alpha_{n-1})$. We construct $Q'$ from $P''$ by cutting off the overlapping segments  $({\bf u}_{i+1}',{\bf u}_{i+2})$ and $({\bf u}_{i+1}',{\bf u}_{i}')$. We apply Lemma~\ref{lemma:straigthening} to $Q'$ thereby obtaining another realization
\[
Q''=({\bf u}_0,\ldots, {\bf u}_{i-1},{\bf u}_i'',{\bf u}_{i+1}'',{\bf u}_{i+3},\ldots, {\bf u}_{n-1}).
\]
   We re-introduce the cut off part to $Q''$ at ${\bf u}_{i+1}''$ as an extension of length $\alpha_{i+2}$  of the segment ${\bf u}_i''{\bf u}_{i+1}''$, whose length in $Q''$ is $\alpha_{i+1}-\alpha_{i+2}>0$, in order to recover a realization of $A$ by the following polygon 
\[
R'=({\bf u}_0,\ldots, {\bf u}_{i-1},{\bf u}_i'',{\bf u}_{i+1}'',{\bf u}_{i+2}'',{\bf u}_{i+3},\ldots, {\bf u}_{n-1}).
\]
 If the turning angle of $Q''$ at ${\bf u}_{i-1}$ equals 0, the same holds for $R'$ and we proceed to Case~2. If the turning angle of $Q''$ at ${\bf u}_{i+1}''$ equals $\pi$, then the turning angle of $R'$ at ${\bf u}_{i+1}''$ equals 0 and we  proceed to Case~2. Finally, if the turning angle of $Q''$ at ${\bf u}_{i+1}''$ equals 0, then $R'$ has a pair of consecutive spurs at ${\bf u}_{i+1}''$ and ${\bf u}_{i+2}''$, that is, a so-called ``crimp.''   We may assume w.l.o.g.\ that  $\alpha_{i+3}<\alpha_{i+1}$. Also we assume that the part $({\bf u}_i'',{\bf u}_{i+1}'',{\bf u}_{i+2}'',{\bf u}_{i+3})$ of $R'$ does not contain a pair of antipodal points, since otherwise we proceed to Case~1.
  Since $({\bf u}_i'',{\bf u}_{i+1}'',{\bf u}_{i+2}'',{\bf u}_{i+3})$ does not contain a pair of antipodal points, $|({\bf u}_i'',{\bf u}_{i+3})|=\alpha_{i+1}+\alpha_{i+3}-\alpha_{i+2}$.
  It follows that

\vspace{-\baselineskip}
\begin{align*}
 |({\bf u}_i'',{\bf u}_{i+3})|+|({\bf u}_i'',{\bf u}_{i+1}'')|+|({\bf u}_{i+1}''{\bf u}_{i+2}'')|+|({\bf u}_{i+2}'',{\bf u}_{i+3})| & =\\
 \alpha_{i+1}+\alpha_{i+3}-\alpha_{i+2}+ \alpha_{i+1}+\alpha_{i+2}+\alpha_{i+3} & = 2(\alpha_{i+1}+\alpha_{i+3}).
\end{align*}
 
If $\alpha_{i+3}+\alpha_{i+1}< \pi$, then the 3 angles 
$\alpha_{i+1}$, $\alpha_{i+2}+\alpha_{i+3}$, and $|({\bf u}_i'',{\bf u}_{i+3})|$ are all  less than $\pi$.
  Moreover, their sum, which is equal to $2(\alpha_{i+3}+\alpha_{i+1})$, is less than $2\pi$, and they satisfy the triangle inequalities. 
  Therefore we can turn the angle at ${\bf u}_{i+2}''$ to 0, by replacing the path  $({\bf u}_i'',{\bf u}_{i+1}'',{\bf u}_{i+2}'',{\bf u}_{i+3})$ on $R'$ by
  a pair of segments of lengths $\alpha_{i+1}$ and $\alpha_{i+2}+\alpha_{i+3}$.
  
Otherwise, $\alpha_{i+3}+\alpha_{i+1}\ge \pi$, and thus,
\[
|({\bf u}_i'',{\bf u}_{i+3})|+|({\bf u}_i'',{\bf u}_{i+1}'')|+|({\bf u}_{i+1}''{\bf u}_{i+2}'')|+|({\bf u}_{i+2}'',{\bf u}_{i+3})|\ge 2\pi.
\]
  In this case, we can apply the induction hypothesis to the closed spherical polygon  
  $({\bf u}_i'',{\bf u}_{i+1}'',{\bf u}_{i+2}'',{\bf u}_{i+3})$. In the resulting realization $S'$, that is w.l.o.g. fixing ${\bf u}_i''$ and ${\bf u}_{i+3}$, we replace the segment $({\bf u}_i'',{\bf u}_{i+3})$ by the remaining part of $R'$ between ${\bf u}_i''$ and ${\bf u}_{i+3}$.
   Let $R''$ denote the resulting realization of $A$.
   If $S'$ is not contained in a great circle then $\mathbf{0}\in \mathrm{int}(\mathrm{conv}(S'))\subseteq \mathrm{int}(\mathrm{conv}(R''))$, and we are done. Otherwise, $S'\setminus ( {\bf u}_{i+3},{\bf u}_{i})$ contains a pair of antipodal points on a half-circle. The same holds for $R''$, and  we proceed to Case~1, which concludes the proof.
 \end{proof}
 
The combination of Theorem~\ref{thm:3d_realization_sphere} with Lemmas~\ref{lemma:3d_realization_char}--\ref{lemma:convex-hull} yields Theorems~\ref{thm:3d_realization_char} and~\ref{thm:3d_realization}. 
The  proof of Lemma~\ref{lemma:convex-hull} can be turned into an algorithm with running time polynomial in $n$ if we assume that every arithmetic operation can be carried out in $O(1)$ time. Nevertheless, we get only a weakly polynomial running time, since we are unable to guarantee a polynomial size encoding of the numerical values that are computed in the process of constructing a spherical polygon realizing $A$
that contains $\mathbf{0}$ in its convex hull in the proof of Lemma~\ref{lemma:convex-hull}.

\section{Crossing Free Realizations in 3D}
\label{sec:3Dcross}

It is perhaps surprising that in $\mathbb{R}^3$ not all realizable angle sequences can be realized  without a crossing. 
The following theorem identifies some angle sequences for which this is the case. They correspond exactly to sequences realizable as a standard musquash~\cite{mushquash}, see Fig.~\ref{fig:mushquash} for an illustration, which is a \emph{thrackle}, that is, a polygon in which every pair of nonadjacent edges cross each other. 

\begin{figure}[htb]
\centering
\includegraphics[scale=0.7]{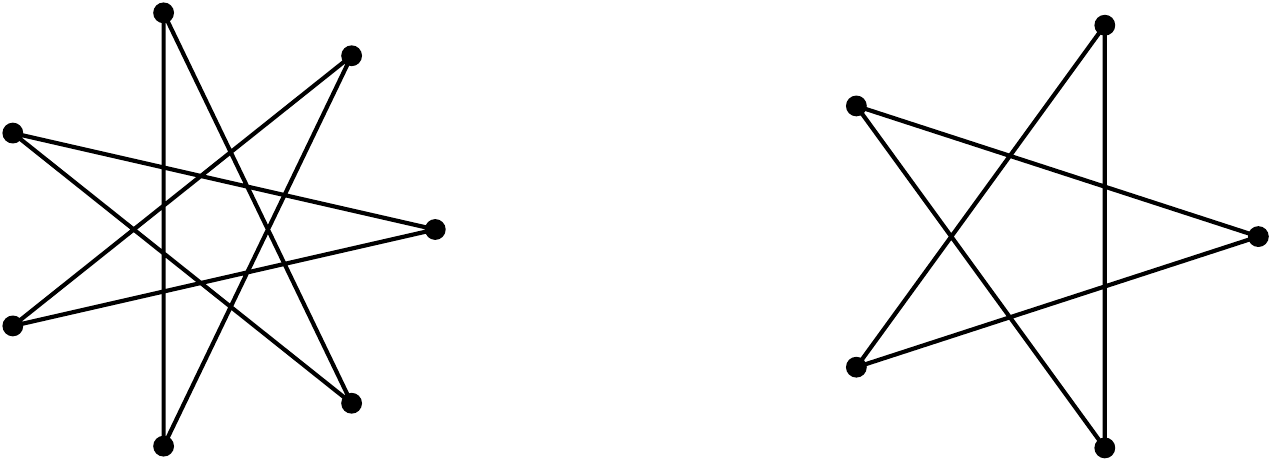}
\caption{Standard musquash with 7 (left) and 5 (right) vertices.}
\label{fig:mushquash}
\end{figure}

\begin{theorem}\label{thm:only_if}
Let  $A=(\alpha_0, \ldots, \alpha_{n-1})$ be an angle sequence, where $n\ge 5$ is odd,  $\sum_{i=0}^{n-1}(\pi-\alpha_i) = \pi$ and $\alpha_i\in(0,\pi)$ for all $i\in \{0,\ldots ,n-1\}$. Then $A$ is realizable in $\mathbb{R}^3$ but every realization lies in an affine plane and has a self-intersection. 
\end{theorem}

\begin{proof}
Let $P'=({\bf u}_0,\ldots, {\bf u}_{n-1})$ be a spherical realization of $A$
corresponding to a realization $P$ in $\mathbb{R}^3$, such that $n\ge 5$ is odd and  $\sum_{i=0}^{n-1}(\pi-\alpha_i) = \pi$.
We permute the vertices of the polygon $P'$ thereby obtaining an auxiliary spherical polygon $Q'=({\bf u}_0,{\bf u}_2, \ldots,{\bf u}_{n-1},{\bf u}_1,{\bf u}_3, \ldots,{\bf u}_0)$. The spherical polygon $Q'$ is well defined since $n$ is odd.

Note that (assuming modulo $n$ arithmetic on the indices) the spherical distance between ${\bf u}_i$ and ${\bf u}_{i+2}$ is bounded above by 
\begin{equation}
\label{eqn:upperii}
    |({\bf u}_i,{\bf u}_{i+2})| \le \pi-\alpha_i+\pi-\alpha_{i-1},
\end{equation} 
for every $i=0,\ldots, n-1$. Indeed,~\eqref{eqn:upperii} is vacuously true if $\alpha_{i}+\alpha_{i+1}\le \pi$. (Recall the definition of $U_i^{2+}$ from Section~\ref{sec:3D}.) Otherwise, $\pi-\alpha_i+\pi-\alpha_{i-1}$ is the spherical distance of any point on $T_i^{2-}$ to ${\bf u}_i$ which is also the farthest possible distance of ${\bf u}_{i+2}$ from ${\bf u}_i$. By~\eqref{eqn:upperii}, the total spherical length of the polygon  $Q'$
is at most $2(n\pi - \sum_{i=0}^{n-1} \alpha_i) =  2\pi$. 

It follows, by applying Fenchel's theorem to $Q'$, that the length of $Q'$ is $2\pi$.
By Lemma~\ref{lemma:3d_realization_char}, we conclude that $Q'$, and thus, also $P'$ are contained in a great circle, which we can assume to be the equator. 
Due to its length, $Q'$ has no self-intersections.
(The polygon $Q$ is in fact convex, but we do not use this in what follows.)
 
Since $P'$ lies in a plane, $P$ also lies in a plane and realizes a signed version
$A^{\pm}$ of the original angle sequence $A$.
As inequality~\eqref{eqn:upperii} must hold with equality due to the length of $Q'$, for all $i=0,\ldots, n$, the all angles in $A^{\pm}$ have the same sign. 
We may assume w.l.o.g.\ that $\alpha_i>0$ for all $i$.
Note that $\sum_{i=0}^{n-1} \alpha_i=(n-1)\pi\ge 4\pi$ by the hypothesis of the theorem. 
Thus, by Theorem~\ref{thm:crossings_number}, the polygon $P$ must have a self-crossing.
\end{proof}

\section{Conclusions}
\label{sec:conclusion}

We devised  efficient algorithms to realize a consistent angle sequence
with the minimum number of crossings in 2D.
In 3D, we can test efficiently whether a given angle sequence is realizable, and find a realization if one exists. 
Every claim we make for $\mathbb{R}^3$ generalizes to $\mathbb{R}^d$, for all $d\ge \mathbb{R}^d$.
The reason is that the circular arcs $I_i$ constructed during an execution of the algorithm in the proof of Theorem~\ref{thm:3d_realization_sphere} depend only on the angles in the sequence, and  would be the same in any higher dimension.

However, it remains an open problem to find an efficient algorithms that computes the minimum number of crossings in generic realizations. 
As we have seen in Section~\ref{sec:3Dcross}, there exist consistent angle sequences in 3D for which every generic realization has crossings.
It is not difficult to see that crossings are unavoidable only if every 3D realization of an angle sequence $A$ is contained in a plane, which is the case, for example, when $A=(\pi-\varepsilon,\ldots, \pi-\varepsilon,(n-1)\varepsilon)$, for odd $n\geq 5$ which is the length of $A$. Thus, an efficient algorithm for this problem would follow by Theorem~\ref{thm:crossings_number}, once one can test efficiently whether $A$ admits a fully 3D realization. 
The evidence that we have points to the following conjecture 
that the converse of Theorem~\ref{thm:only_if} also holds.

\begin{conjecture}\label{con:1}
An angle sequence $A=(\alpha_0, \ldots, \alpha_{n-1})$, where $\alpha_i\in(0,\pi)$ and $n\ge 4$, that can be realized by a polygon in $\mathbb{R}^3$, has a realization by a self-intersection free polygon in $\mathbb{R}^3$ if and only if  $n$ is even or $\sum_{i=0}^{n-1}(\pi-\alpha_i)\not = \pi$. 
\end{conjecture}

It can be seen that Conjecture~\ref{con:1} is equivalent to the claim that every realization $A$ in $\mathbb{R}^3$ has a self-intersection if and only if $A$ can be realized in $\mathbb{R}^2$ as a \emph{thrackle}.

Can our results in $\mathbb{R}^2$ or $\mathbb{R}^3$ be  extended to   broader interesting classes of graphs?
A natural analog of our problem in $\mathbb{R}^3$ would be a construction of triangulated spheres with prescribed dihedral angles, discussed in a recent paper by Amenta and Rojas~\cite{AmentaR20}. For convex polyhedra, Mazzeo and Montcouquiol~\cite{MM11} proved, settling Stoker's conjecture, that dihedral angles determine face angles. 


\bibliographystyle{plain}
\bibliography{references}

\begin{thebibliography}{10}

\bibitem{AmentaR20}
Nina Amenta and Carlos Rojas.
\newblock Dihedral deformation and rigidity.
\newblock {\em Computational Geometry: Theory and Applications}, 90:101657,
  2020.

\bibitem{BekosFK19}
Michael~A. Bekos, Henry F{\"{o}}rster, and Michael Kaufmann.
\newblock On smooth orthogonal and octilinear drawings: Relations, complexity
  and {K}andinsky drawings.
\newblock {\em Algorithmica}, 81(5):2046--2071, 2019.

\bibitem{BS02_spherical}
Richard~E. Buckman and Nicholas Schmitt.
\newblock Spherical polygons and unitarization.
\newblock {\em Preprint}, 2002.
\newblock http://www.gang.umass.edu/reu/2002/polygon.html.

\bibitem{CJK20}
Katie Clinch, Bill Jackson, and Peter Keevash.
\newblock Global rigidity of direction-length frameworks.
\newblock {\em Journal of Combinatorial Theory, Series B}, 145:145--168, 2020.

\bibitem{DiBattistaKLLW12}
Giuseppe {Di~Battista}, Ethan Kim, Giuseppe Liotta, Anna Lubiw, and Sue
  Whitesides.
\newblock The shape of orthogonal cycles in three dimensions.
\newblock {\em Discrete \& Computational Geometry}, 47(3):461--491, 2012.

\bibitem{di1996angles}
Giuseppe Di~Battista and Luca Vismara.
\newblock Angles of planar triangular graphs.
\newblock {\em SIAM Journal on Discrete Mathematics}, 9(3):349--359, 1996.

\bibitem{dmw-11-pda}
Yann Disser, Mat\'u\v{s} Mihal\'ak, and Peter Widmayer.
\newblock A polygon is determined by its angles.
\newblock {\em Computational Geometry: Theory and Applications}, 44:418--426,
  2011.

\bibitem{driscoll1998numerical}
Tobin~A Driscoll and Stephen~A Vavasis.
\newblock Numerical conformal mapping using cross-ratios and {D}elaunay
  triangulation.
\newblock {\em SIAM Journal on Scientific Computing}, 19(6):1783--1803, 1998.

\bibitem{F51_curve}
Werner Fenchel.
\newblock On the differential geometry of closed space curves.
\newblock {\em Bulletin of the American Mathematical Society}, 57(1):44--54,
  1951.

\bibitem{G98_anglegraphs}
Ashim Garg.
\newblock New results on drawing angle graphs.
\newblock {\em Computational Geometry: Theory and Applications}, 9(1-2):43--82,
  1998.

\bibitem{grunbaum1990rotation}
Branko Gr{\"u}nbaum and Geoffrey~Colin Shephard.
\newblock Rotation and winding numbers for planar polygons and curves.
\newblock {\em Transactions of the American Mathematical Society},
  322(1):169--187, 1990.

\bibitem{GHS00_morph}
Leonidas Guibas, John Hershberger, and Subhash Suri.
\newblock Morphing simple polygons.
\newblock {\em Discrete \& Computational Geometry}, 24(1):1--34, 2000.

\bibitem{JacksonJ10}
Bill Jackson and Tibor Jord{\'{a}}n.
\newblock Globally rigid circuits of the direction-length rigidity matroid.
\newblock {\em Journal of Combinatorial Theory, Series {B}}, 100(1):1--22,
  2010.

\bibitem{JacksonK11a}
Bill Jackson and Peter Keevash.
\newblock Necessary conditions for the global rigidity of direction-length
  frameworks.
\newblock {\em Discrete \& Computational Geometry}, 46(1):72--85, 2011.

\bibitem{JohnS09}
Audrey~Lee{-}St. John and Ileana Streinu.
\newblock Angular rigidity in {3D}: {C}combinatorial characterizations and
  algorithms.
\newblock In {\em Proc. 21st Canadian Conference on Computational Geometry
  ({CCCG})}, pages 67--70, 2009.

\bibitem{KM1999}
Michael Kapovich and John~J. Millson.
\newblock On the moduli space of a spherical polygonal linkage.
\newblock {\em Canadian Mathematical Bulletin}, 42:307--320, 1999.

\bibitem{matousek2013lectures}
Ji\v{r}\'{\i} Matou\v{s}ek.
\newblock {\em Lectures on Discrete Geometry}, volume 212 of {\em Graduate
  Texts in Mathematics}.
\newblock Springer-Verlag New York, 2002.

\bibitem{MM11}
Rafe Mazzeo and Gr\'egoire Montcouquiol.
\newblock Infinitesimal rigidity of cone-manifolds and the {S}toker problem for
  hyperbolic and {E}uclidean polyhedra.
\newblock {\em Journal of Differential Geometry}, 87(3):525--576, 2011.

\bibitem{mushquash}
Grace Misereh and Yuri Nikolayevsky.
\newblock Thrackles containing a standard musquash.
\newblock {\em Australasian Journal Of Combinatorics}, 70(2):168--–184, 2018.

\bibitem{PaninaS10}
Gaiane Panina and Ileana Streinu.
\newblock Flattening single-vertex origami: The non-expansive case.
\newblock {\em Computational Geometry: Theory and Applications},
  43(8):678--687, 2010.

\bibitem{Patrignani08}
Maurizio Patrignani.
\newblock Complexity results for three-dimensional orthogonal graph drawing.
\newblock {\em J. Discrete Algorithms}, 6(1):140--161, 2008.

\bibitem{SaliolaW04}
Franco Saliola and Walter Whiteley.
\newblock Constraining plane configurations in {CAD:} circles, lines, and
  angles in the plane.
\newblock {\em {SIAM} Journal on Discrete Mathematics}, 18(2):246--271, 2004.

\bibitem{S99_crossratios}
Jack Snoeyink.
\newblock Cross-ratios and angles determine a polygon.
\newblock {\em Discrete \& Computational Geometry}, 22(4):619--631, 1999.

\bibitem{StreinuW04}
Ileana Streinu and Walter Whiteley.
\newblock Single-vertex origami and spherical expansive motions.
\newblock In {\em on Japanese Conference on Discrete and Computational
  Geometry}, volume 3742 of {\em LNCS}, pages 161--173. Springer, 2004.

\bibitem{Sullivan2008}
John~M. Sullivan.
\newblock Curves of finite total curvature.
\newblock In Alexander~I. Bobenko, John~M. Sullivan, Peter Schr{\"o}der, and
  G{\"u}nter~M. Ziegler, editors, {\em Discrete Differential Geometry}, pages
  137--161. Birkh{\"a}user, Basel, 2008.

\bibitem{thomassen1980planarity}
Carsten Thomassen.
\newblock Planarity and duality of finite and infinite graphs.
\newblock {\em Journal of Combinatorial Theory, Series B}, 29:244--271, 1980.

\bibitem{tutte1963draw}
William~Thomas Tutte.
\newblock How to draw a graph.
\newblock {\em Proceedings of the London Mathematical Society}, 3(1):743--767,
  1963.

\bibitem{vijayan1986geometry}
Gopalakrishnan Vijayan.
\newblock Geometry of planar graphs with angles.
\newblock In {\em Proc. 2nd {ACM} Symposium on Computational Geometry}, pages
  116--124, 1986.

\bibitem{Wiener64}
Christian Wiener.
\newblock {\em \"Uber Vielecke und Vielflache}.
\newblock Teubner, Leipzig, 1864.

\bibitem{zayer_etal:05}
Rhaleb Zayer, Christian R{\"o}ssler, and Hans-Peter Seidel.
\newblock Variations on angle based flattening.
\newblock In N.~A. Dodgson, M.~S. Floater, and M.~A. Sabin, editors, {\em
  Advances in Multiresolution for Geometric Modelling}, Mathematics and
  Visualization, pages 187--199. Springer, 2005.

\end{thebibliography}

\end{document}